%% file: main_arxiv.tex
\documentclass[11pt]{article}

\usepackage{geometry}
\usepackage{amsmath,amsthm,amssymb}
\usepackage[utf8]{inputenc}
\usepackage{fullpage}
\usepackage{url}
\usepackage{graphicx}
\usepackage{soul,xcolor}
\usepackage{tikz}
\usepackage{float}
\usepackage[normalem]{ulem}
\usepackage[capitalise]{cleveref}
\usetikzlibrary{decorations.pathreplacing}

\usepackage{charsty}
\usepackage{subcaption}
\usepackage[ruled,vlined]{algorithm2e}

\newtheorem{definition}{Definition}
\newtheorem{theorem}{Theorem}
\newtheorem{lemma}{Lemma}

\title{Fairness in the Assignment Problem with Uncertain Priorities\footnote{This work is supported by NSF grant CCF-2113798.}}

\newcommand*\samethanks[1][\value{footnote}]{\footnotemark[#1]}
\author{Zeyu Shen\thanks{Computer Science Department, Duke University, Durham, NC 27708-0129. Email: \texttt{\{zeyu.shen,zhiyi.wang,xingyu.zhu\}@duke.edu}, \texttt{\{btfain,kamesh\}@cs.duke.edu}.} \and Zhiyi Wang\samethanks[1] \and Xingyu Zhu\samethanks[1] \and Brandon Fain\samethanks[1] \and Kamesh Munagala\samethanks[1]}
\date{}

\begin{document}

\maketitle

\begin{abstract}
    \input{abstract}
\end{abstract}

\input{introduction.tex}

\input{preliminaries.tex}

\input{desiderata.tex}


\input{results.tex}

\input{positive_results}

\input{experiments}

\input{conclusion}

\newpage

\bibliographystyle{plain}
\bibliography{refs.bib}

\end{document}

%% file: abstract.tex
In the assignment problem, a set of items must be allocated to unit-demand agents who express ordinal preferences (rankings) over the items. In the assignment problem with priorities, agents with higher priority are entitled to their preferred goods with respect to lower priority agents. A priority can be naturally represented as a ranking and an uncertain priority as a distribution over rankings. For example, this models the problem of assigning student applicants to university seats or job applicants to job openings when the admitting body is uncertain about the true priority over applicants. This uncertainty can express the possibility of bias in the generation of the priority ranking. We believe we are the first to explicitly formulate and study the assignment problem with uncertain priorities. We introduce two natural notions of fairness in this problem: stochastic envy-freeness (SEF) and likelihood envy-freeness (LEF). We show that SEF and LEF are incompatible and that LEF is incompatible with ordinal efficiency. We describe two algorithms, Cycle Elimination (CE) and Unit-Time Eating (UTE) that satisfy ordinal efficiency (a form of ex-ante Pareto optimality) and SEF; the well known random serial dictatorship algorithm satisfies LEF and the weaker efficiency guarantee of ex-post Pareto optimality. We also show that CE satisfies a relaxation of LEF that we term 1-LEF which applies only to certain comparisons of priority, while UTE satisfies a version of proportional allocations with ranks. We conclude by demonstrating how a mediator can model a problem of school admission in the face of bias as an assignment problem with uncertain priority.

%% file: introduction.tex
\section{Introduction}

Consider a motivating example of the assignment problem where a number of university admission slots (the items) must be assigned to student applicants (the agents). The university slots could be at a single university or several. Applicants might have preferences over different universities, or might have preferences over different slots at the same university (for example, some slots might be associated with merit-based financial aid, or include admission to particular academic programs). Applicants are \textit{unit-demand}, meaning they only need to be assigned a single slot (and derive no benefit from being assigned multiple).

Most university systems employ some form of priority-based admissions; this can be expressed through a ranking over applicants. For example, a priority might rank applicants by standardized exam scores, or perhaps by some more complex holistic assessment. Given any deterministic priority (a ranking), one might naturally  solve the assignment problem using the \textit{serial dictatorship} rule, so that students choose their most preferred remaining university slot one at a time in order of their standardized exam score. Indeed, systems roughly like this are employed in several countries around the world such as the Indian Institutes of Technology~\cite{JET-IIT}.

Despite the appeal of such a simple and ostensibly fair system, there is reason to suspect that any scoring or ranking system is based on imperfect noisy signals of the true underlying priority (whatever that might be). For example, an applicant A scoring 1 point higher on a standardized exam or holistic assessment than another applicant B is not, in general, 100\% more likely to be a better student than B. Even more worryingly, studies show that standardized exam performance is closely related to demographic factors such as race and income~\cite{DEM2013}, leading to uncertainty based on social bias and inequality in addition to random noise like whether one had a good breakfast the day of an exam. More holistic assessments are further vulnerable to the well documented phenomenon of implicit bias against historically marginalized groups~\cite{ImplicitBias}. Ignoring these uncertainties may result in arbitrary decisions (deterministically preferring one applicant over another when the comparison is unclear and noisy) and systemic discrimination against historically marginalized groups.

Previous work has attempted to solve the second problem of bias without explicitly modeling an uncertain priority by adapting the so-called ``Rooney Rule''~\cite{kleinberg2018, celis2020interventions}. There are variations, but roughly speaking these methods reserve a number of ``minority'' spots and prioritize this many ``minority'' applicants in some serial dictatorship assignment. This approach can lead to \textit{fairness gerrymandering}~\cite{Kearns18} by which structured subgroups remain disadvantaged. In particular, Rooney Rule style approaches are predicated on a single binary distinction of the applicant population into ``majority'' (or privileged) and ``minority'' (or disadvantaged) applicants. But in reality, applicant identity is multidimensional (race, gender, income, disability, first language, etc.) and bias  can compound along intersections. In fact, it is perfectly plausible that the vast \textit{majority} of applicants are disadvantaged (that is, suffer from bias leading to underestimation of their priority) along one or more dimensions of identity, though not all to the same extent. In addition to group identity, there may sometimes be uncertainties related to the priority of individual applicants, unique circumstances that merit accounting. 

For these reasons, we consider the more general problem that takes as input an uncertain priority, expressed as a probability distribution over rankings of applicants. The generality of the input to our algorithms ensures that a decision maker can fully model the complexity of uncertainty and bias inherent in the creation of a priority. This modeling problem is outside the scope of this paper, though we do provide an example for our experiments in Section~\ref{sec:experiments}. Rather, our emphasis is on the question of characterizing fairness and efficiency given a random priority, and providing algorithms to compute random assignments that satisfy these desiderata.

\subsection{Contributions}
We study an extension of the random assignment problem~\cite{MB2001, MB2002, Abdulkadiroglu2003OrdinalEA} in which a decision maker must allocate a number of items to unit-demand agents in a way that is consistent with an \textit{uncertain priority} represented as a distribution over rankings of the agents. To the best of our knowledge, we are the first to characterize this more general problem. 

In general we want to compute a random assignment that is simultaneously \textit{efficient} with respect to agent preferences over the items and \textit{fair} with respect to the agent priorities. \textit{Ordinal efficiency} (OE)~\cite{MB2001} generalizes the concept of Pareto efficiency to the case of a random assignment. Our main contribution is to characterize two alternative notions of fairness for the random assignment problem with uncertain priorities in Section~\ref{sec:desiderata}. The first notion, which we call \textit{stochastic envy-freeness} (SEF), guarantees that any agent whose priority first-order stochastically dominates another agent's priority should prefer their own (random) assignment to that of the other agent. The second notion, which we call \textit{likelihood envy-freeness} (LEF), guarantees that the likelihood (over the random assignment) that an agent prefers the assignment of another should be at most the likelihood (over the uncertain priority) that the latter agent has higher priority than the former.  

We introduce additional notions that helps more finely distinguish between algorithms that satisfy one of the above notions. The first is a relaxation of LEF called 1-LEF that holds only when an agent has higher priority than another with probability 1.  The next is {\em ranked proportionality} (PROP), where the allocation of any agent should stochastically dominate the allocation where she gets her $i$-th preferred item with probability $p_i$ if she herself is ranked at position $i$ with that probability. 

Formal definitions are provided in Section~\ref{sec:desiderata}. We provide illustrative examples of these concepts as well as justification for why multiple definitions of fairness might be appropriate in Section~\ref{sec:need}. 

In Section~\ref{impossible} we show that it is impossible to guarantee OE and LEF simultaneously. We also show that it is impossible to guarantee SEF and LEF simultaneously. Given this, we focus on achieving OE and SEF. In Section~\ref{sec:algorithms} we describe two algorithms: \textit{Unit-time Eating} (UTE) and \textit{Cycle Elimination} (CE). We show that both of these algorithms satisfy OE and SEF. To more finely distinguish between these algorithms, we show that CE also satisfies the relaxed 1-LEF property, while UTE satisfies PROP. We also show that any algorithm achieving OE cannot achieve PROP and 1-LEF simultaneously, so that we cannot achieve a super-set of the properties achieved by these algorithms. 

It is straightforward to observe that the well known \textit{Random Serial Dictatorship} (RSD) that samples a priority from $\Sigma$ and then uses the serial dictatorship satisfies LEF, PROP, and is ex-post Pareto efficient, though it does not satisfy OE~\cite{MB2001}. We obtain a nearly complete characterization of achievable subsets of our efficiency and fairness properties, as shown in Table~\ref{tab0}.

\begin{table}[htbp]
\begin{center}
\begin{tabular}{||c || c |c | c | c |c ||} 
 \hline
 Algorithm & OE & SEF & LEF & 1-LEF & PROP \\
 \hline\hline
 RSD &  &  & \textbf{\checkmark} & \textbf{\checkmark} & \textbf{\checkmark} \\ 
 \hline
 UTE (new) & \textbf{\checkmark} & \textbf{\checkmark} &  & & \textbf{\checkmark}  \\
 \hline
 CE (new) & \textbf{\checkmark} & \textbf{\checkmark} &  & \textbf{\checkmark} & \\
 \hline
 \hline
\end{tabular}
\caption{\label{tab0} Summary of fairness properties achieved.}
\end{center}
\end{table}

In Section~\ref{sec:experiments} we return to a consideration of our motivating application of biased school admissions. We provide a practical example modeling an uncertain priority in the presence of bias and compare our CE and UTE algorithms with previous approaches to address bias using ``Rooney Rule'' style approaches~\cite{kleinberg2018, celis2020interventions}.

\subsection{Related Work}

\paragraph{Random Assignment.} There is a large body of work studying the problem of random assignment with no priority (or, in our framework, when the priority is uniform). The work of~\cite{randomserialdictatorship} proposed a {\em random serial dictatorship} mechanism, which draws an ordering of agents uniformly at random and let them choose items in that order, and showed that this mechanism is ex-post efficient. The work of~\cite{Zhou1990OnAC} observed that though random serial dictatorship is fair, it is not efficient when the agents are endowed with Von Neumann-Morgenstern preferences over lotteries. The work of~\cite{MB2001} introduced a notion of efficiency that is stronger than ex-post efficiency, namely {\em ordinal efficiency}, and showed that random serial dictatorship is not ordinally efficient. They proposed the {\em probabilistic serial} rule that is ordinally efficient. Moreover, probabilistic serial is (stochastically) envy-free while random serial dictatorship is not. The work of~\cite{Abdulkadiroglu2003OrdinalEA} studied the relationship between ex-post efficiency and ordinal efficiency, showing that a lottery induces an ordinally efficient random assignment if and only if each subset of the full support of the lottery is undominated (in a specific sense).

Subsequent works investigated natural extensions of the canonical setup. The work of~\cite{MB2002} considered the problem of random assignment in the case where agents can opt out, and characterised probabilistic serial by ordinal efficiency, envy-freeness, strategyproofness, and equal treatment of equals in this setting. The work of~\cite{rankefficiency} studied the notion of {\em rank efficiency}, which maximises the number of agents matched to their first choices. 

\paragraph{Fair Ranking.} The assignment problem with priority is closely related to the subset selection problem that has been studied extensively as a problem in fair ranking~\cite{kleinberg2018, CME19, NCGW20, celis2020interventions, EGGL20, MC21, GB21} where the goal is to optimize some latent measure of utility for the algorithm designer subject to group fairness constraints on the resulting ranking. Recent work considers explicitly modeling the uncertainty from bias when estimating a ranking based on observed utilities~\cite{SKJ21}, similar to our approach in modeling an uncertain priority. Our work differs from the fair ranking literature in that we study a more general assignment problem in which agents may not all have the same preferences over items. Of course, one can always translate a given ranking into an assignment by employing the serial dictatorship rule, but this need not be ordinally efficient~\cite{MB2001}. Instead, we formulate our desiderata more explicitly in the wider context of the assignment problem itself.  

\paragraph{Two-sided matching.} School choice problems are often studied in the context of two-sided matching, where applicants have preferences over schools and schools have preferences over applicants. For example, the deferred acceptance algorithm (and its extensions) calculates stable matchings and has been extensively studied and deployed in the real world~\cite{gale1962college, Roth82, Roth84, Atila03, Atila05}. Our problem is different in two ways. First, the ``items'' in our problem (eg., school seats) share a single common priority over applicants, so the notion of stability simply means no applicant of lower priority is assigned an item preferred by an agent of higher priority. However, our setting is more complex in the second sense: The shared priority is uncertain, and the assignment will be random, requiring an extension of existing fairness properties and algorithms.


%% file: preliminaries.tex
\section{Preliminaries}
We are given $n$ unit demand agents $\mathcal{A} = \{1, 2, \hdots, n\}$ and a set of $m$ items $\mathcal{I}$. We assume without loss of generality that $m \geq n$ (if not, one can create additional ``dummy'' items that are least preferred by all agents). We write $a \succ_i b$ to denote that agent $i$ prefers item $a$ to item $b$. Each agent has ordinal preferences represented as a total order over $\mathcal{I}$, that is, for every agent $i$ we have a permutation $\pi_i: \mathcal{I} \rightarrow \{1, \hdots, n\}$ such that $\pi_i(a) < \pi_i(b)$ if and only if $a \succ_i b$.\footnote{In general, results extend trivially to the case where agents may have \textit{objective} indifferences between items, meaning that if any agent is indifferent between two items then all agents are indifferent between those items. 
However, our results do \textit{not} necessarily extend straightforwardly if agents have subjective indifferences, see~\cite{MB2001}.}

A \textit{simple priority} over agents is a permutation $\sigma: \mathcal{A} \rightarrow \{1, \hdots, n\}$ where $\sigma(i) < \sigma(j)$ means that $i$ has higher priority than $j$. A \textit{random priority} is a probability distribution over simple priorities which we denote as $\Sigma = \{(\sigma_k, \rho_k)\}$ where each $\sigma_k$ is a simple priority, $\rho_k \geq 0$, and $\sum_k \rho_k = 1$.

A \textit{simple assignment} is a matching $f: \mathcal{A} \rightarrow \mathcal{I}$. A \textit{lottery} is a probability distribution over simple assignments which we denote as $\mathcal{L} = \{(f_k, p_k)\}$ where each $f_k$ is a simple assignment, $p_k \geq 0$, and $\sum_k p_k = 1$.

Following~\cite{MB2001}, we call a probability distribution over $[m]$ itself a \textit{random allocation} to an agent. It is important to note that agents have ordinal preferences over deterministic items which only induces a partial order over random allocations. That is, given $\pi_i$, it may be unclear whether $i$ would prefer one random allocation to another. We denote by $P = \{p_{ij}\}$ a \textit{random assignment}, the $n$ by $m$ matrix where $P_{i}$, the $i$-th row, is agent $i$'s random allocation, and where $\sum_i p_{ij} = 1$ for all columns $j$. In general, a random assignment $P$ can be induced by one or more lotteries, the existence of which is guaranteed by the Birkhoff-von Neumann Theorem, but a particular lottery induces a unique random assignment $P$.

In the \textit{assignment problem with uncertain priorities} we are given a random priority $\Sigma$ and agent preferences $\{\pi_i\}$ and we must compute a random assignment.

%% file: desiderata.tex
\section{Desiderata}
\label{sec:desiderata}

In this section we introduce the normative properties that an algorithm for the random assignment with uncertain priorities problem should satisfy. Broadly speaking, these desiderata require that the algorithm be efficient with respect to agent preferences and fair with respect to agent priorities.

\subsection{Efficiency}
A simple assignment $f$ is \textit{Pareto efficient} (or Pareto optimal) if it is not dominated by any other simple assignment, which simply means that there is no alternative such that no agent is worse off and at least one agent is better off. 

\begin{definition}[Pareto Efficiency]
A simple assignment $f$ is Pareto efficient if for all simple assignments $g$ one of the following holds: (i) $\exists i \in \mathcal{A}$ such that $f(i) \succ_i g(i)$, or (ii) $g(i) \nsucc f(i)$ for all $i \in [n]$. 
\end{definition}

A lottery $\mathcal{L}$ is ex-post Pareto efficient if every simple assignment in the support of $\mathcal{L}$ (i.e., every simple assignment $f_k$ with $p_k > 0$) is Pareto efficient. 

A stronger efficiency property for a random assignment is \textit{ordinal efficiency} {\sc (OE)}~\cite{MB2001}. To define ordinal efficiency we must first define the notion of stochastic dominance.

\begin{definition}[Stochastic Domination]
A probability distribution $X$ \textit{stochastically dominates} another distribution $Y$ under permutation $\pi$ (denoted $X \succ^{sd}_{\pi} Y$) if for all $t \in \{1, \hdots, n\}$ it holds that $\sum_{r=1}^{t} X_{\pi^{-1}(r)} \geq \sum_{r=1}^{t} Y_{\pi^{-1}(r)}$, where $\pi^{-1}$ is the inverse permutation. A random assignment $P$ is stochastically dominated by a random assignment $Q \neq P$ if the random allocation induced by $Q$ stochastically dominates the random allocation induced by $P$ under preferences $\pi_i$ for every agent $i \in [n]$.
\end{definition}

Note that this implies the following: If random assignment $Q$ stochastically dominates random assignment $P$, then every agent prefers $Q$ to $P$ under any Von Neumann-Morgenstern utility function consistent with their ordinal preferences. Now we can define ordinal efficiency, following~\cite{MB2001}.

\begin{definition}[Ordinal Efficiency, OE]
We say that a random assignment $P$ is \textit{ordinally efficient} if it is not stochastically dominated by any other random assignment.
\end{definition}

At a high level, a random assignment is ordinally efficient if there is no other random assignment that is better for all agents and all utility functions consistent with their ordinal preferences. The property is not trivial: Some natural algorithms such as random serial dictatorship are Pareto efficient but not ordinally efficient.   

\subsection{Fairness}
\label{fairness}

We define fairness in terms of envy. We say that one agent \textit{envies} another if the former prefers the item assigned to the latter. Envy of a lower priority agent constitutes a justified complaint against an assignment; ideally we would like to compute an \textit{envy-free} assignment with respect to the priority.   

\begin{definition}[Envy-Freeness]
We say that a simple assignment $f$ is \textit{envy-free} with respect to a simple priority $\sigma$ if for all $i, j \in [n]$, $\sigma(i) < \sigma(j) \implies f(i) \succ_i f(j).$
\end{definition}

However, it is immediately evident that it is impossible to compute a single simple assignment that is envy-free in this sense for every simple priority in the support of a random priority (for example, if there are two agents with uncertain priority who both prefer the same item). Instead, we need to compute a random assignment so that each agent is fairly treated ex-ante (for example, so that each agent has a fair probability of receiving the preferred good).   

There are two natural ways to generalize the concept of envy to a random assignment with a random priority. One is to imagine that one agent envies another if the random allocation of the latter stochastically dominates that of the former under the former's ordinal preferences. Envy of this type forms a justified complaint if the envying agent also stochastically dominates the envied agent in terms of the random priority. More formally, 

\begin{definition}[Stochastic Envy-Freeness, SEF]
\label{def_sef}
Consider a random assignment $P$ generated under a random priority $\Sigma$. Let $S_i$ be the probability distribution over $[n]$ induced by $\Sigma$ for agent $i$, that is, for $r \in [n]$, $S_{ir} = \sum_{k: \sigma_k(i)=r} \rho_k$. Let $\sigma_{*}$ be the identity permutation, i.e., $\sigma_{*}(i) = i$. $P$ is \textit{stochastically envy-free} (SEF) with respect to $\Sigma$ if for all $i, j \in [n]$, $S_i \succ^{sd}_{\sigma_{*}} S_j \implies P_i \succ^{sd}_{\pi_i} P_j.$
\end{definition}

Loosely speaking, the implication of stochastic envy-freeness can be read as ``if agent $i$ probably has higher priority than $j$ then $i$ should prefer their random allocation to $j$'s under all utility functions consistent with $i$'s ordinal preferences.''

A second way to generalize envy is by considering the likelihood of envy (in the simple sense) with respect to a lottery inducing a given random assignment. Envy of this type is justified if the likelihood of agent $i$ envying another agent $j$ is greater than the likelihood over the random priority that $i$ has lower priority than agent $j$. We call a random assignment \textit{likelihood envy-free} if there is a lottery which induces it and has no envy of this kind.

\begin{definition}[Likelihood Envy-Freeness, LEF]
A random assignment $P$ satisfies \textit{likelihood envy-freeness} (LEF)  under  $\Sigma$ if $P$ can be induced by a lottery $\mathcal{L}$ such that for all $i, j \in [n]$, $\Pr_{\sigma \sim \Sigma}[\sigma(i) < \sigma(j)] \leq \Pr_{f \sim \mathcal{L}}[f(i) \succ_i f(j)].$
\end{definition}

In other words, LEF means that an agent $i$ who is $\ell$-likely to have higher priority than another agent $j$ should be at least $\ell$-likely to prefer their assigned item to $j$'s. 

We say an algorithm satisfies OE (resp. SEF, LEF) if it always produces random assignment that satisfies OE (resp. SEF, LEF). As we show in Section~\ref{impossible}, it is not possible to guarantee SEF and LEF simultaneously.

\subsection{Relationship between LEF and SEF}
\label{sec:need}
The relationship between SEF and LEF is subtle; neither implies the other and it is not immediately evident which is the ``better'' or more ``natural'' fairness property. We present two examples to illustrate that an assignment satisfying only one of SEF and LEF might still be unfair, so that both properties are useful competing notions of fairness, and neither is strictly stronger than the other.

We first present an example which shows that an assignment that satisfies SEF can be unfair. Consider $n = 2$ agents and $m = 2$ items which we label $a, b$ for clarity. Both agents prefer $a$ to $b$, and the random priority is simply $\Sigma = \{(\sigma, 1)\}$ with $\sigma(1) < \sigma(2)$, {\em i.e.} agent 1 has higher priority than agent 2 with probability $1$. In this setup, allocating $\frac{1}{2}$ unit of $a$ and $b$ to both agent yields an assignment that satisfies SEF. However, this assignment is clearly unfair, because even though agent 1 has higher priority than agent 2, they are getting the same assignment. Notice that this assignment does not satisfy LEF. In this instance, LEF could be used to characterize how much one agent is prioritized over the other.

The next example shows that an assignment that only satisfies LEF can also be unfair. Consider $n = 2$ agents and $m = 100$ items which we label $i_1, \ldots, i_{100}$ for clarity. The preferences of both agents are $i_1 \succ \cdots \succ i_{100}$. The random priority is given by $\Sigma = \{(\sigma_1, \frac{1}{2}), (\sigma_2, \frac{1}{2})\}$ with $\sigma_1(1) < \sigma_1(2)$ and $\sigma_2(2) < \sigma_2(1)$. In other words, both agents have the same priority. In this setup, allocating $\frac{1}{2}$ unit of $i_1$ and $\frac{1}{2}$ unit of $i_{100}$ to agent 1 and $\frac{1}{2}$ unit of $i_{99}$ and $\frac{1}{2}$ unit of $i_{100}$ to agent 2 yields an assignment that satisfies LEF. Notice that this assignment can be induced by a lottery $\mathcal{L} = \{(f_1, \frac{1}{2}), (f_2, \frac{1}{2})\}$ where $f_1(1) = i_1, f_1(2) = i_{100}, f_2(1) = i_{100}, f_2(2) = i_{99}$. However, this assignment is clearly unfair, because even though the two agents have the same priority, agent 1 gets a strictly better assignment than agent 2. This shows that LEF alone has limitations as well, and the appropriate concept here is SEF.

The above examples show that SEF and LEF provide reasonable competing notions of fairness. When combined with the efficiency notion of OE, we will show in Section~\ref{impossible} that LEF and OE are incompatible. If OE is replaced by the weaker notion of Pareto-efficiency, then it is easy to check that random serial dictatorship (RSD), which simply samples a priority of agents from the distribution and allocates each agent their favorite remaining item in this priority order,  satisfies LEF\footnote{To see why RSD satisfies LEF, suppose the random priority is given by $\Sigma = \{(\sigma_k, \rho_k)\}$, then the random assignment produced by RSD can be induced by the lottery $\mathcal{L} = \{(f_k, \rho_k)\}$, where $f_k$ is the deterministic assignment produced by letting agents successively choose an item based on the order given by $\sigma_k$.} and pareto efficiency. Thus, in our work, we will focus on the more non-trivial part of finding algorithms that satisfy SEF and OE. 




\subsection{Additional Fairness Criteria}
As we show in Section~\ref{sec:algorithms}, there can be multiple algorithms that satisfy the same subset of the fairness criteria. We therefore consider two additional notions to more finely distinguish between them. 

The first criterion is the following relaxation of LEF: If agent $i$ with probability $1$ has higher priority than another agent $j$ then agent $i$ should certainly (again, with probability 1) not envy $j$.

\begin{definition}[1-LEF]
A random assignment $P$ under random priority $\Sigma$ satisfies {\sc 1-LEF} if there exists some lottery $\mathcal{L}$ which induces $P$ such that for all agents $i \neq j \in [n]$, if $\Pr_{\sigma \sim \Sigma}[\sigma(i) < \sigma(j)] = 1$, then $\Pr_{f \sim \mathcal{L}}[f(i) \succ_i f(j)] = 1$.
\end{definition}

The next criterion is called {\em Ranked Proportionality} (PROP), which captures stochastic dominance over an allocation that matches the probability an agent gets her $i^{th}$ ranked item to the probability of she being ranked at position $i$. Note that if all rankings of agents were equally likely, this captures stochastic dominance to an allocation that assigns every item to every agent uniformly at random.

\begin{definition}[{\sc PROP}] Given a random priority $\Sigma = \{(\sigma_k, \rho_k)\}$, we define the baseline allocation $\overline{P}_i$ for agent $i$ by $\overline{P}_{i\pi_i^{-1}(r)} = S_{ir} = \sum_{k: \sigma_k(i) = r}\rho_k$ for all $r \in [n]$. In other words, if an agent $i$ ranks the $r$-th in the random priorities with probability $p$, then we add $p$ fraction of the $r$-th preferred item of agent $i$ to her baseline allocation. For an allocation to satisfy {\em ranked proportionality} ({\sc PROP}), it should stochastically dominate this baseline for each agent. 
\end{definition}

%% file: results.tex
\section{Impossibility Results}
\label{impossible}

In this part, we present several impossibility results. We note that these are existential hardness results, not computational. We begin by observing that LEF is incompatible with OE. 

\begin{theorem}
\label{DOMOE}
 LEF is incompatible with OE.
\end{theorem}
\begin{proof}
We present an instance in which no random assignment can satisfy both LEF and OE. There are $n=4$ agents and $m=4$ items which we label $a, b, c, d$ for clarity. Agent preferences are given by
$$\pi_1, \pi_3: a \succ b \succ c \succ d, \quad \pi_2, \pi_4: b \succ a \succ c \succ d$$
Moreover, we consider the priority $\Sigma = \{(\sigma_1, \frac{1}{2}), (\sigma_2, \frac{1}{2})\}$ where
$$\sigma_1(4) < \sigma_1(2) < \sigma_1(3) < \sigma_1(1),$$
$$\sigma_2(3) < \sigma_2(1) < \sigma_2(4) < \sigma_2(2).$$
In other words, with probability $\frac{1}{2}$ under $\sigma_1$, agent 4 has the highest priority, then agent 2, then agent 3, finally agent 1. Similarly for $\sigma_2$ with probability $\frac{1}{2}$. Assume for contradiction that there exists a random assignment $P = [p_{ij}]$, together with a lottery $\mathcal{L}$ which induces $P$, satisfying LEF and OE. By definition of LEF, we note that
$$\Pr_{f \sim \mathcal{L}}[f(3) \succ_3 f(1)] \geq \Pr_{\sigma \sim \Sigma}[\sigma(3) < \sigma(1)] = 1,$$
so it must be that $\Pr_{f \sim \mathcal{L}}[f(3) \succ_3 f(1)] = 1$. Thus, we must have $p_{1a} = 0$, because otherwise there would exist a simple assignment in the lottery in which agent $1$ is assigned with $a$ and agent $3$ is assigned with some less preferred item under $\pi_3$. By the same reasoning, we note that $p_{2b} = 0$.

Also by definition of LEF, observe that
$$\Pr_{f \sim \mathcal{L}}[f(2) \succ_2 f(3)] \geq \Pr_{\sigma \sim \Sigma}[\sigma(2) < \sigma(3)] = \frac{1}{2}.$$
This implies $p_{3a} < 1$, as otherwise we  have $f(3) = a$ for all $f \sim \mathcal{L}$; combined with the fact that $p_{2b} = 0$, we would have $f(3) \succ_2 f(2)$ for all $f \sim \mathcal{L}$, which contradicts $\Pr_{f \sim \mathcal{L}}[f(2) \succ_2 f(3)] \geq \frac{1}{2}$. 

Since $p_{1a} = 0$, $p_{3a} < 1$, and $\sum_{i}p_{ia} = 1$, it follows that $p_{2a} + p_{4a} > 0$. Similarly, we have $p_{1b} + p_{3b} > 0$. Without loss of generality, we assume that $p_{2a} > 0$ and $p_{1b} > 0$ (if $p_{4a} > 0$ or $p_{3b} > 0$, the proof proceeds similarly). Let $p_{min} = \min(p_{2a}, p_{1b})$; define random assignment $Q = [q_{ij}]$ by
$$q_{ij} = 
\begin{cases}
p_{ij} \text{ if } i \notin \{1, 2\} \text{ and } j \notin \{a, b\}\\
p_{ij} + p_{min} \text{ if } (i, j) = (1, a) \text{ or } (2, b)\\
p_{ij} - p_{min} \text{ if } (i, j) = (1, b) \text{ or } (2, a)
\end{cases}
$$

We can see that $Q$ stochastically dominates $P$. In particular, all that is different in $Q$ is that agent 1 swaps agent 2 some of agent 2's allocated probability mass on item $a$ in exchange for an equivalent amount of agent 1's probability mass on item $b$. Since $a \succ_1 b$ and $b \succ_2 a$ and nothing else changes, agents 1 and 2 prefer $Q$, and nothing has changed for agents 3 and 4. This contradicts with the fact that $P$ satisfies OE. Thus, we can conclude that no random assignment in this instance satisfies LEF and OE.
\end{proof}

Theorem~\ref{DOMOE} can be interpreted as a fundamental tradeoff between efficiency and fairness conceived as LEF. Next, we show that LEF and SEF are two fundamentally different notions of fairness that are incompatible with one another. As we will see later in Section~\ref{sec:algorithms}, each of LEF and SEF independently can be guaranteed. Thus, neither notion of fairness is subsumed by the other. 

\begin{theorem}
\label{theorem2}
 LEF is incompatible with SEF.
\end{theorem}
\begin{proof}
We present an instance in which no random assignment can satisfy both LEF and SEF. There are $n=5$ agents and $m=5$ items which we label $a, b, c, d, e$ for clarity. Preferences are given by
\begin{align*}
&\pi_1, \pi_3: a \succ b \succ c \succ d \succ e, \quad \pi_2, \pi_4: b \succ a \succ c \succ d \succ e,\\
&\pi_5: a \succ c \succ b \succ d \succ e. 
\end{align*}
We consider the priority $\Sigma = \{(\sigma_1, \frac{1}{2}), (\sigma_2, \frac{1}{2})\}$ defined by
$$\sigma_1(3) < \sigma_1(5) < \sigma_1(1) < \sigma_1(4) < \sigma_1(2),$$
$$\sigma_2(4) < \sigma_2(5) < \sigma_2(2) < \sigma_2(3) < \sigma_2(1).$$

In other words, with probability $\frac{1}{2}$ under $\sigma_1$, agent 3 has the highest priority, then agents 5, 1, 4, and finally 2. Similarly for $\sigma_2$.

Assume for contradiction that there exists a random assignment $P = [p_{ij}]$, together with a lottery $\mathcal{L}$ which induces $P$, that satisfies LEF and SEF. Since agent 3 always has higher priority than agent 1 and agent 3 prefers $a$ over all other items, LEF implies that $p_{1a} = 0$. Similarly, since agent 4 always has higher priority than agent 2 prefers $b$ over all other itmes, LEF implies that $p_{2b} = 0$.

Recall that $S_i$ is the probability density over $[n]$ induced by $\Sigma$ for agent $i$ and $\sigma^*$ is the identity permutation. Since $S_1 \succ_{\sigma^*}^{sd} S_2$ by construction and $P$ satisfies {\sc (SEF)} by assumption, we have $P_1 \succ_{\pi_1}^{sd} P_2$. Combined with the fact that $p_{1a} = 0$, we must have $p_{2a} = 0$. Similarly, $p_{1b} = 0$.

We next show $p_{1c} = p_{2c} = \frac{1}{2}$. First, observe that LEF guarantees
$$\Pr_{f \sim \mathcal{L}}[f(4) \succ_4 f(2)] \geq \Pr_{\sigma \sim \Sigma}[\sigma(4) < \sigma(2)] = 1,$$
Thus, since $e$ is the least preferred item by agent 4, we must have $p_{4e} = 0$. Also by LEF, we have
$$\Pr_{f \sim \mathcal{L}}[f(1) \succ_1 f(4)] \geq \Pr_{\sigma \sim\Sigma}[\sigma(1) < \sigma(4)] = \frac{1}{2},$$
i.e. $\Pr_{f \sim \mathcal{L}}[f(1) \succ_1 f(4)] \geq \frac{1}{2}$. On the other hand, since $p_{4e} = 0$, the worst item that agent 4 can get under $\pi_4$ is $d$, so
$$\Pr_{f \sim \mathcal{L}}[f(1) \succ_1 f(4)] \leq p_{1a} + p_{1b} + p_{1c} = p_{1c},$$ since we earlier found that $p_{1a} = p_{1b} = 0$.
 Recall $\Pr_{f \sim \mathcal{L}}[f(1) \succ_1 f(4)] \geq \frac{1}{2}$, we get $p_{1c} \geq \frac{1}{2}$. Similarly, we have $p_{2c} \geq \frac{1}{2}$. Since $\sum_{i}p_{ic} = 1$, it must be the case that $p_{1c} = p_{2c} = \frac{1}{2}$. We deduce that for any $f \sim \mathcal{L}$, either $f(1) = c$ or $f(2) = c$, because on one hand, for any fixed $f$, we should have $f(1) \neq f(2)$, while on the other hand, $p_{1c} + p_{2c} = 1$.

Observe that $p_{5a} \leq \frac{1}{2}$. This follows directly from LEF, because
$$\Pr_{f \sim \mathcal{L}}[f(3) \succ_3 f(5)] \geq \Pr_{\sigma \sim \Sigma}[\sigma(3) < \sigma(5)] = \frac{1}{2};$$
if $p_{5a} > \frac{1}{2}$, we would have $\Pr_{f \sim \mathcal{L}}[f(3) \succ_3 f(5)] < 1 - p_{5a} = \frac{1}{2}$, leading to contradiction. What's more, we have $p_{5c} = 0$, since we already have $p_{1c} + p_{2c} = 1$.

On one hand, we should have $\Pr_{f \sim \mathcal{L}}[f(5) \succ_5 f(1)$ and $f(5) \succ_5 f(2)] = 1,$ since $\sigma_i(5) < \sigma_i(1)$ and $\sigma_i(5) < \sigma_i(2)$ for $i \in \{1, 2\}$; but on the other hand, we have
$\Pr_{f \sim \mathcal{L}}[f(5) \succ_5 f(1)$ and  $f(5) \succ_5 f(2)] \leq p_{5a} \leq \frac{1}{2},$
because for any $f \sim \mathcal{L}$, either $f(1) = c$ or $f(2) = c$, so $f(5) \succ_5 f(1)$ and $f(5) \succ_5 f(2)$ if and only if $f(5) = a$. This leads to contradiction. Thus, LEF and SEF are incompatible.
\end{proof}

We finally show that OE, 1-LEF, and PROP are simultaneously incompatible. This will inform the design of algorithms in Section~\ref{sec:algorithms}.

\begin{lemma}
\label{lem:imposs3}
There is an instance where no allocation simultaneously satisfies OE, 1-LEF, and PROP.
\end{lemma}
\begin{proof}
We use the instance in the proof of Theorem~\ref{DOMOE}. Since agent $3$ is ranked higher than agent $1$ with probability $1$ and since their preferences over items is identical, 1-LEF implies agent $1$ is allocated item $a$ with probability $0$. Now PROP implies agent $1$ receives item $b$ with probability at least $1/2$. By a similar reasoning, agent $2$ must get item $a$ with probability at least $1/2$. But any such allocation cannot be OE, completing the proof.
\end{proof}

%% file: positive_results.tex
\section{Algorithms}
\label{sec:algorithms}
As we have seen, LEF is a very strong notion of fairness which is incompatible with both OE and SEF. In the following, we present two algorithms -- {\em cycle elimination} (CE) and {\em unit time eating} (UTE) -- that satisfy both OE and SEF. In addition, we show that CE satisfies 1-LEF and UTE satisfies PROP. Given Lemma~\ref{lem:imposs3}, we cannot design an algorithm that achieves OE and both these properties.

Therefore,  both CE and UTE are reasonable fair allocation algorithms in that they satisfy efficiency (OE) and envy-freeness (SEF). The choice of which to implement depends on whether we care more about a form of proportionality in the resulting allocation (UTE satisfies PROP) or whether we care about additional envy-freeness in a deterministic sense (CE satisfies 1-LEF).



\subsection{Cycle Elimination algorithm}
We first introduce a \emph{Cycle Elimination algorithm} (CE), which works by constructing a directed graph based on the random priority, and allocate items based on this graph.

To begin with, we introduce the \emph{Probabilistic Serial rule}~\cite{MB2001}, a continuous algorithm which works as follows. Initially, each agent $i$ goes to their favorite item $j$ and starts ``eating'' it (that is, increasing $p_{ij}$) at unit speed. It is possible that several agents eat the same item at the same time. Whenever an item is fully eaten, each of the agents eating it goes to their favorite remaining item not fully allocated (that is, $\sum_{i} p_{ij} < 1$) and starts eating it in the same way. This process continues until all items are consumed, or all the agents are full (that is, $\sum_{j} p_{ij} = 1$). We use {\sc PS($\mathcal{A}$, $\mathcal{I}$)} to denote the assignment produced by running Probabilistic Serial rule on the set of agents $\mathcal{A}$ and items $\mathcal{I}$. 


We construct a graph from $\Sigma$, which we call a Stochastic-Dominance graph (SD-graph), as follows: Start with a graph with $n$ vertices, where the $i$-th vertex corresponds to the $i$-th agent. For any pair of distinct agents $i$ and $j$, if $S_i \succ_{\sigma^*}^{sd} S_j$, then we draw a directed edge from $i$ to $j$. The algorithm is now formally stated in Algorithm~\ref{cycle}.

\begin{algorithm}[htbp]
\SetAlgoLined
{\bf Input:} Set of agents $\mathcal{A}$, set of items $\mathcal{I}$, SD-graph $G$\;
Let $\hat{G}$ be the condensation\footnotemark of $G$\;
Let $\widetilde{\mathcal{A}}$ be the set of agents that belong to a strongly connected component whose in-degree in $\hat{G}$ is zero\;

 \eIf{$\mathcal{A} = \widetilde{\mathcal{A}}$}{
   Output {\sc PS($\mathcal{A}$, $\mathcal{I}$)}\;
 }
 {
    $\mathcal{A}' \leftarrow \mathcal{A} \setminus \widetilde{\mathcal{A}}$; $\mathcal{I}'\leftarrow \mathcal{I}\: \setminus$ {\sc PS($\widetilde{\mathcal{A}}$; $\mathcal{I}$)}; $G' \leftarrow G \setminus \widetilde{\mathcal{A}}$\;
    Output {\sc PS($\widetilde{\mathcal{A}}$, $\mathcal{I}$)} + Eliminate($\mathcal{A}'$, $\mathcal{I}'$, $G'$)\;
 }
 \caption{Cycle Elimination, Eliminate($\mathcal{A}$, $\mathcal{I}$, $G$)}
 \nllabel{cycle}
\end{algorithm}
\footnotetext{Condensation of a graph is a directed acyclic graph formed by contracting each strongly connected component to a single vertex.}

\paragraph{Analysis.} Our main result is the following theorem.

\begin{theorem}
\label{thm:main1}
The Cycle Elimination algorithm satisfies OE, SEF, and 1-LEF. It runs in $O(n^3 + nm + n|\Sigma|)$ time.
\end{theorem}

\begin{proof}
Theorem 1 in~\cite{MB2001} states that any simultaneous eating algorithm where each agent always eats from her favorite remaining item satisfies OE. Hence, CE satisfies OE.

To show SEF, fix two agents $i$ and $j$, and assume $S_i \succ^{sd}_{\sigma_{*}} S_j$. Let $P$ be the random assignment produced by CE. We show that $P_i \succ^{sd}_{\pi_i} P_j$. Since $S_i \succ^{sd}_{\sigma_{*}} S_j$, there exists an edge from $i$ to $j$ in the SD-graph. Thus, $i$ and $j$ either belong to the same strongly connected component, or the strongly connected component of $i$ has higher topological order than that of $j$'s. Either way, we have $P_i \succ^{sd}_{\pi_i} P_j$, from which we can conclude that CE satisfies SEF.

To show 1-LEF, fix two agents $i$ and $j$, and assume $\Pr_{\sigma \sim \Sigma}[\sigma(i) < \sigma(j)] = 1$. Let $P$ be the random assignment produced by CE. We show that, for any lottery $\mathcal{L}$ inducing $P$, we have $\Pr_{f \sim \mathcal{L}}[f(i) \succ_i f(j)] = 1$. We use proof by contradiction. Assume that there exists a lottery $\mathcal{L}_0$ which induces $P$ such that $\Pr_{f \sim \mathcal{L}_0}[f(i) \succ_i f(j)] <  1$. This implies that there exists two items $a$ and $b$ such that $a \succ_i b$, $p_{ib} > 0$ and $p_{ja} > 0$. On the other hand, since $\Pr_{\sigma \sim \Sigma}[\sigma(i) < \sigma(j)] = 1$, so we have $S_i \succ^{sd}_{\sigma_{*}} S_j$; thus, the strongly connected component that agent $i$ belongs to must have higher topological order than that of $j$'s. By CE, agent $j$ could start eating only when agent $i$ is completely full. Thus, $p_{ja} > 0$ implies that there is still item $a$ remaining when agent $i$ finishes eating; this leads to contradiction, because $a\succ_i b$ implies that $i$ could eat $a$ instead of $b$. Therefore, we must have $\Pr_{f \sim \mathcal{L}}[f(i) \succ_i f(j)] = 1$ for all lotteries $\mathcal{L}$ which induces $P$, from which we can conclude that CE satisfies 1-LEF.

To show the running time, preprocessing $\Sigma$ in order to compute the stochastic dominance relation between agents takes $O(n|\Sigma| + n^3)$ time. Constructing the SD-graph by the stochastic dominance relation between agents takes $O(n^2)$ time, as there are $n\choose{2}$ pair of agents. Given the SD-graph, running CE takes $O(nm)$ time. This is because we only need to consider at most $m$ time points: the time at which each item is eaten up. We divide this process into $m$ time intervals. During each time interval, each agent keeps eating the same item, so it simply takes $O(n)$ time to keep track of the state of each agent, and the running time over $m$ intervals is $O(nm)$. Hence, the total running time is $O(n^3 + nm + n|\Sigma|)$.
\end{proof}

\subsection{Unit-time Eating Algorithm}
We next introduce the \emph{Unit-time Eating Algorithm} (UTE). Recall that $\Sigma = \{(\sigma_k, \rho_k)\}$. Essentially, the algorithm works by dividing the time into $n$ units, each of duration one; in time unit $t$, the $t$-th ranked agent in $\sigma_k$ eats their favorite item among those left over at rate $\rho_k$ for all $k$. The procedure is formally stated in Algorithm~\ref{unit}.

\begin{algorithm}[htbp]
\SetAlgoLined

\For{$t = 1, \ldots, n$}{
The $t$-th ranked agent in each $\sigma_k$ eats their favorite item among those left over at rate $\rho_k$ for all $(\sigma_k, \rho_k) \in \Sigma$\;
}
 \caption{Unit-time Eating Algorithm}
 \nllabel{unit}
\end{algorithm}

\paragraph{Analysis.} We show the following theorem.

\begin{theorem}
\label{thm:main2}
The Unit-time Eating Algorithm satisfies OE, SEF, and PROP. Further, it runs in $O(n^2|\Sigma| + nm)$ time.
\end{theorem}

\begin{proof}
By Theorem 1 in~\cite{MB2001}, we have UTE satisfies OE.

To show SEF, fix two agents $i$ and $j$; assume that $S_i \succ^{sd}_{\sigma_{*}} S_j$. Let $P$ be the random assignment produced by UTE, we show that $P_i \succ^{sd}_{\pi_i} P_j$. Let $t_k$ be the time when item $\pi_i^{-1}(k)$ has been eaten up. Fix some $k \in [m]$; because $S_i \succ^{sd}_{\sigma_{*}} S_j$, we have
$$    \sum_{t = 1}^{\lfloor t_k\rfloor} S_{it} \geq \sum_{t = 1}^{\lfloor t_k\rfloor} S_{jt} \qquad  \mbox{and} \qquad   \sum_{t = 1}^{\lceil t_k\rceil} S_{it} \geq \sum_{t = 1}^{\lceil t_k\rceil} S_{jt} $$



Combining these gives
\begin{equation}
\label{eq3}
\big(t_k - \lfloor t_k\rfloor\big)S_{i\lceil t_k\rceil} + \sum_{t = 1}^{\lfloor t_k\rfloor} S_{it} \geq \big(t_k - \lfloor t_k\rfloor\big)S_{j\lceil t_k\rceil} + \sum_{t = 1}^{\lfloor t_k\rfloor} S_{jt}.    
\end{equation}

Observe that
$$\sum_{r = 1}^k P_{i\pi_i^{-1}(r)} = \big(t_k - \lfloor t_k\rfloor\big)S_{i\lceil t_k\rceil} + \sum_{t = 1}^{\lfloor t_k\rfloor} S_{it},$$
$$\sum_{r = 1}^k P_{j\pi_i^{-1}(r)} \leq \big(t_k - \lfloor t_k\rfloor\big)S_{j\lceil t_k\rceil} + \sum_{t = 1}^{\lfloor t_k\rfloor} S_{jt},$$
which gives $\sum_{r = 1}^k P_{i\pi_i^{-1}(r)} \geq \sum_{r = 1}^k P_{j\pi_i^{-1}(r)}$. Because this holds for all $k \in [m]$, we conclude that $P_i \succ^{sd}_{\pi_i} P_j$. Hence, UTE satisfies SEF.

To show PROP, fix some agent $i$. Suppose the allocation produced by UTE for this agent is $P_i$, and the baseline allocation for this agent is $\overline{P}_i$. We will show that $P_i \succ_{\pi_i}^{sd} \overline{P}_i$. Let $t_k$ be the time when item $\pi_i^{-1}(k)$ has been eaten. Fix some $k \in [m]$. Clearly, we have $t_k \geq k$, because in order to eat up $\pi_i^{-1}(k)$, we have to eat up $\pi^{-1}(r)$ for all $r < k$. We observe that
$$\sum_{r = 1}^k P_{i\pi_i^{-1}(r)} = \big(t_k - \lfloor t_k\rfloor\big)S_{i\lceil t_k\rceil} + \sum_{t = 1}^{\lfloor t_k\rfloor} S_{it}.$$
Combined with $t_k \geq k$, we have 
$$\sum_{r = 1}^k P_{i\pi_i^{-1}(r)} \geq \sum_{t = 1}^{k} S_{it} = \sum_{r = 1}^k\overline{P}_{i\pi_i^{-1}(r)}.$$
Because this holds for all $k \in [m]$, we conclude that that $P_i \succ_{\pi_i}^{sd} \overline{P}_i$, and hence UTE satisfies PROP.

To show running time, preprocessing $\Sigma$ to obtain the eating speed of each agent in each unit time interval takes $O(n^2|\Sigma|)$ time. Then, running UTE takes $O(nm)$ time, as we similarly only need to consider at most $m$ time points and keeping track of the state of each agent in each time interval takes $O(n)$ time. Therefore, the total running time is $O(n^2|\Sigma| + nm)$.
\end{proof}

%% file: experiments.tex
\section{Generating Random Priorities and Empirical Results}
\label{sec:experiments}

In this section, we will demonstrate how one could obtain random priorities in practical settings using an example of school admission under implicit bias. In several environments based on such generative model, we will compare our proposed algorithms, namely Cycle Elimination (CE) and Unit-time Eating (UTE), with other common bias mitigating allocation algorithms such as ``the Rooney Rule''~\cite{celis2020interventions}. We empirically demonstrate that all existing algorithms induce stochastic envy. To show this, using the same notation as in Definition~\ref{def_sef}, we say a pair of agents $i$ and $j$ form a {\em stochastic envy pair} if $Z_i\succ_{\sigma_*}^{sd} Z_j$ but $P_i \nsucc_{\pi_i}^{sd} P_j$, and we will count the number of stochastic envy pairs produced by each algorithm.

\newcommand{\major}{advantaged\ }
\newcommand{\minor}{disadvantaged\ }

\subsection{Random Priority in School Admission}
\label{sec:experiments-priority-gen}
Consider a group of $N$ students, including $n$ \minor students with indices $\cbr{1,\dots, n}$ and $N-n$ \major students with indices $\cbr{n+1, \ldots, N}$. Suppose that they are competing for admission priorities of $\ell$ schools with capacities $c_1,\dots, c_\ell\in \NN$ that $\sum_{i=1}^\ell c_i = N$, in which process disadvantaged students are subjected to implicit bias on their capability. We will quantify the effect of implicit bias in the experiments. 
This is equivalent to allocating $N$ items to $N$ agents, where items correspond to seats, and the agents' preferences are their school choices. 

Denote the $j$-th seat of school $i$ as $s_{i_j}$, then the set of seats is $S\triangleq \bigcup_{i=1}^\ell\{s_{i_1}, s_{i_2},\dots, s_{i_{c_i}}\}$. For any ordinal preference $\tilde\pi:[\ell]\to[\ell]$ over the schools, it induces an ordinal preference $\pi:S\to[N]$ over the seats such that for any $s_{i_j}\in S$, $\pi(s_{i_j}) = j + \sum_{k:\tilde\pi(k)<\tilde\pi(i)}c_k.$ In other words, if a student prefers school $a$ to school $b$, then all seats of school $a$ are preferred over the seats of school $b$. For the seats in the same school, smaller indices are preferred. In the following, we describe how random priorities over students are generated.


For each student, we assign a ``capability score" $x_i$ that is drawn from the same distribution $\cD$, and students with higher capability score should have higher priority. Moreover, assume every student from the \minor group is subjected to a multiplicative implicit bias $b_i$, which is independently sampled from some distribution $\cB$. A \minor student with capability score $x_i$ is perceived to have a biased score $\hat x_i \triangleq b_ix_i$. We will also consider {\em additive bias} $\hat x_i \triangleq x_i+b_i$ in our experiments. The admission committee make decisions based on the perceived scores (which are biased for disadvantaged students and equal to the true scores for advantaged students). 

\begin{figure*}
\centering
\begin{tabular}{cc|cccccc}
\hline
\hline
$\ell$ & $\beta$ & N     & RN    & R   & RR  & CE & UTE \\ \hline
    & 0.2     & 3.4   & 0     & 3.4 & 10  & 0  & 0   \\
1   & 0.5     & 1.2   & 0     & 1.2 & 10  & 0  & 0   \\
    & 0.8     & 0.6   & 0     & 0.6 & 10  & 0  & 0   \\ \hline
    & 0.2     & 14.3  & 42.8  & 2.6 & 3.8 & 0  & 0   \\
2   & 0.5     & 14.5  & 42.8  & 1.0   & 3.8 & 0  & 0   \\
    & 0.8     & 19.7  & 42.8  & 0.6 & 3.8 & 0  & 0   \\ \hline
    & 0.2     & 88.8  & 175.7 & 1.6 & 2.5 & 0  & 0   \\
3   & 0.5     & 98.9  & 175.7 & 0.7 & 2.5 & 0  & 0   \\
    & 0.8     & 103.5 & 175.7 & 0.4 & 2.5 & 0  & 0   \\ \hline
\hline
\end{tabular}
\qquad
\begin{tabular}{cc|cccccc}
\hline
\hline
$\ell$ & $\beta$ & N     & RN    & R    & RR   & CE & UTE \\ \hline
         & 0.2       & 7.0   & 0     & 15.4 & 25.4 & 0  & 0   \\
1        & 0.5       & 7.3   & 0     & 6.2  & 36.9 & 0  & 0   \\
         & 0.8       & 7.8   & 0     & 2.8  & 42.2 & 0  & 0   \\ \hline
         & 0.2       & 38.2  & 36.2  & 11.2 & 16.8 & 0  & 0   \\
2        & 0.5       & 37.0  & 38.3  & 4.5  & 22.9 & 0  & 0   \\
         & 0.8       & 37.4  & 39.6  & 2.2  & 25.6 & 0  & 0   \\ \hline
         & 0.2       & 156.2 & 183.3 & 7.3  & 9.8  & 0  & 0   \\
3        & 0.5       & 141.4 & 200.5 & 3.4  & 12.6 & 0  & 0   \\
         & 0.8       & 127.6 & 205.5 & 1.9  & 15.5 & 0  & 0   \\ \hline
\hline
\end{tabular}
    
\caption{\label{tab:add}\label{tab:mult} Number of stochastic envy pairs under multiplicative bias (left) and additive bias (right).} 
\end{figure*}


For each experiment, we fix a set of unbiased capability scores $\cbr{x_i}_{i = 1}^N$ for the students, where $x_i\overset{\mathrm{iid}}{\sim} \cD$. Then, we take $n$ bias parameters $\{b_i\}_{i = 1}^n$ independently from $\cB$. The perceived scores of the students are $\cbr{\hat x_1,\dots, \hat x_{n}, x_{n+1}, \dots, x_N}$, where $\hat{x_i} = b_ix_i$. Now imagine we are the admission committee. We know $\cB$, $\cD$, and the perceived scores of the students. The goal is to approximately recover the underlying true scores of the students. To do this, we compute a posterior distribution for the bias factor of each disadvantaged student given $\cB$, the biased score of this student, and $\cD$. Concretely, the density of the posterior distribution for the bias factor of the $i^{\text{th}}$ disadvantaged student, which we denote by $\bm{b_i}$, is $f_{\bm{b_i}}(b) = \frac{f_{\cB}(b)f_{\cD}(\hat{x}_i/b)}{\int_{0}^\infty f_{\cB}(u)f_{\cD}(\hat{x}_i/u)du}$. Given $\{\bm{b_i}\}_{i = 1}^n$, we independently draw $q$ sets of bias parameters for disadvantaged students, where we denote the $j^{\text{th}}$ set of bias parameters as $\{b_i^{(j)}\}_{i = 1}^n$, {\em i.e.} $b_i^{(j)}\overset{\mathrm{iid}}{\sim} \bm{b_i}$. Let the ordinal relationship induced by $\{b_i^{(j)}\}_{i = 1}^n$ be $\sigma^{(j)}$. We consider the random priority $\{(\sigma^{(j)}, \frac{1}{q})\}_{j = 1}^q$. We denote the random priority induced by $q$ sets of bias parameters as $\Sigma^{(q)}$.

\subsection{Algorithms for Comparison}
To compare with CE and UTE, we consider four alternative solutions to the allocation problem under implicit bias. Fix a set of biased scores $\cbr{\hat x_1,\dots, \hat x_{n}, x_{n+1}, \dots, x_N}$, let $\hat\sigma$ denote its induced ordinal relationship. For a deterministic priority $\sigma$ over the students and ordinal preferences $\Pi\triangleq\{\pi_i\}_{i\in [N]}$ of students over the seats, let $GS(\sigma, \Pi)$ denote the deterministic assignment produced by the Gale-Shapley algorithm~\cite{gale1962college} which produces a stable matching between students and seats.

The algorithms that we compare with are as follows: 
\begin{enumerate}
    \item \emph{Naive Stable Matching} (N) takes deterministic  priority $\hat\sigma$ and returns the assignment $P_{\text{N}}(\hat\sigma,\Pi)\triangleq GS(\hat\sigma,\Pi)$. 
    \item \emph{Random Naive Stable Matching} (RN) takes the random priority $\Sigma^{(q)} = \cbr{(\sigma_i, p_i)}_{i = 1}^q$ and outputs a lottery based on (N), namely $ \cbr{(P_{\text{N}}(\sigma_i,\Pi), p_i)}_{i = 1}^q$.
    \item \emph{Rooney Stable Matching} (R) takes in the deterministic priority $\hat\sigma$ as input. Using the Rooney constraint in Theorem 3.3 of~\cite{celis2020interventions}, it  creates a new priority $\hat\sigma_{\text{R}}$. We present this formally in Algorithm~\ref{alg:rooney-like}.
Using $\hat\sigma_{\text{R}}$, \emph{Rooney Stable Matching} returns the assignment $P_{\text{R}}(\hat{\sigma}_{\text{R}},\Pi)\triangleq GS(\hat\sigma_{\text{R}}, \Pi)$.

\item \emph{Random Rooney Stable Matching} (RR) takes the random priority $\Sigma^{(q)} = \cbr{(\sigma_i, p_i)}_{i = 1}^q$  and outputs a lottery based on (R), namely $ \cbr{(P_{\text{R}}(\sigma_i,\Pi), p_i)}_{i = 1}^q$.
\end{enumerate}

\begin{algorithm}[htbp]
\SetAlgoLined
Let $A,B$ be the ordered sub-sequences of \minor and \major candidates in $\hat\sigma$ respectively, {\em i.e.} $p<q\iff \hat\sigma(A[p])<\hat\sigma(A[q])$\;
$i,j\gets 0$\;
\While{$i+j<N$}{
    \eIf{$\lfloor \frac{i}{i + j}\rfloor < \frac{n}{N} \normalfont{\textbf{ or }} \hat\sigma(i) < \hat\sigma(n+j)$}
    {$\hat\sigma_{\text{R}}(A[i]) = i+j$ and    $i\gets i+1$\;}
    {$\hat\sigma_{\text{R}}(B[j]) = i+j$ and    $j\gets j+1$\;}
}
\Return $\hat\sigma_{\text{R}}$
\caption{Proportional Rooney-rule-like Constraint~\cite{celis2020interventions}}
\nllabel{Ronnie}
\label{alg:rooney-like}
\end{algorithm}

\subsection{Prevalence of Stochastic Envy}
We now demonstrate that with random priority induced by the generative model described in Section \ref{sec:experiments-priority-gen}, stochastic envy exists for the bias mitigating algorithms N, RN, R, RR.

We consider an admission problem with $\ell$ schools each with $\lfloor \frac{N}{\ell + 1} \rfloor$ seats. Every student $i\in [N]$ has a uniformly random preference order over the $\ell$ schools. There is also a "dummy school" with $N-\ell\lfloor \frac{N}{\ell + 1} \rfloor$ seats representing no admission. Every student prefers seats in the dummy school the least. For seats in the same school, all students have the same preference order. This represents the situation in which schools may distribute educational resources to students based on their rank when admitted.

We take $N = 35$ and $n=10$, and experiment with $\ell = 1, 2, 3$. For each choice of $k$, we experiment with $\beta = 0.2, 0.5, 0.8$. For multiplicative bias, we take $\mathcal{D} = \text{Exponential}(1)$ and $\cB = \text{Exponential}(\beta)$; for additive bias, we take $\cD=\text{Uniform}(0,2)$ and $\cB=\text{Uniform}(0,\beta)$. 
Figure~\ref{tab:mult} presents  the number of stochastic envy pairs for each algorithm averaged over 100 experiments.  For each experiment, the random priority is computed with 1000 sets of bias parameters.  

Except for RN in the one school setting
, stochastic envy exists in all other scenarios for N, RN, R, RR. While empirically Rooney-rule-like constraints do significantly reduce the number of stochastic envy pairs compared to applying no mitigation mechanism at all, we still need CE or UTE to obtain guaranteed SEF.

%% file: conclusion.tex
\section{Conclusion}
We conclude with some open questions. First, even though SEF and LEF are incompatible, we do not know whether they can be compatible under certain natural generative assumptions on the random priorities and agent preferences. Second, it is known~\cite{MB2001} that OE (and hence CE and UTE) is incompatible with strategyproofness under natural assumptions. However, it is interesting to explore whether SEF alone is also incompatible with strategy-proofness. Finally, can our framework be extended to the scenario where the agent preferences are random as well, {\em i.e.} each agent reports a distribution over preferences instead of a deterministic preference?

%% file: main_arxiv.bbl
\begin{thebibliography}{10}

\bibitem{randomserialdictatorship}
Atila Abdulkadiroglu and Tayfun S{\"o}nmez.
\newblock Random serial dictatorship and the core from random endowments in
  house allocation problems.
\newblock {\em Econometrica}, 66:689--701, 1998.

\bibitem{Abdulkadiroglu2003OrdinalEA}
Atila Abdulkadiroglu and Tayfun S{\"o}nmez.
\newblock Ordinal efficiency and dominated sets of assignments.
\newblock {\em J. Econ. Theory}, 112:157--172, 2003.

\bibitem{Atila05}
Atila Abdulkadiroğlu, Parag~A. Pathak, and Alvin~E. Roth.
\newblock The new york city high school match.
\newblock {\em American Economic Review}, 95(2):364--367, May 2005.

\bibitem{Atila03}
Atila Abdulkadiroğlu and Tayfun Sönmez.
\newblock School choice: A mechanism design approach.
\newblock {\em American Economic Review}, 93(3):729--747, June 2003.

\bibitem{ImplicitBias}
Marianne Bertrand and Sendhil Mullainathan.
\newblock Are emily and greg more employable than lakisha and jamal? a field
  experiment on labor market discrimination.
\newblock {\em American Economic Review}, 94(4):991--1013, September 2004.

\bibitem{MB2001}
Anna Bogomolnaia and Hervé Moulin.
\newblock A new solution to the random assignment problem.
\newblock {\em Journal of Economic Theory}, 100(2):295--328, 2001.

\bibitem{celis2020interventions}
L~Elisa Celis, Anay Mehrotra, and Nisheeth~K Vishnoi.
\newblock Interventions for ranking in the presence of implicit bias.
\newblock In {\em Proceedings of the 2020 Conference on Fairness,
  Accountability, and Transparency}, pages 369--380, 2020.

\bibitem{DEM2013}
Ezekiel Dixon-Roman, Howard Everson, and John Mcardle.
\newblock Race, poverty and sat scores: Modeling the influences of family
  income on black and white high school students' sat performance.
\newblock {\em Teachers College Record}, 115, 05 2013.

\bibitem{EGGL20}
Vitalii Emelianov, Nicolas Gast, Krishna~P. Gummadi, and Patrick Loiseau.
\newblock On fair selection in the presence of implicit variance.
\newblock EC '20, page 649–675, New York, NY, USA, 2020. Association for
  Computing Machinery.

\bibitem{rankefficiency}
Clayton~R Featherstone.
\newblock Rank efficiency: Modeling a common policymaker objective.
\newblock Technical report, The Wharton School, University of Pennsylvania,
  2020.

\bibitem{gale1962college}
David Gale and Lloyd~S Shapley.
\newblock College admissions and the stability of marriage.
\newblock {\em The American Mathematical Monthly}, 69(1):9--15, 1962.

\bibitem{GB21}
David Garc\'{\i}a-Soriano and Francesco Bonchi.
\newblock Maxmin-fair ranking: Individual fairness under group-fairness
  constraints.
\newblock KDD '21, page 436–446, New York, NY, USA, 2021. Association for
  Computing Machinery.

\bibitem{JET-IIT}
JEE.
\newblock Joint entrance examination (2011) report.
\newblock Technical report, Indian Institutes of Technology, 2011.

\bibitem{Kearns18}
Michael Kearns, Seth Neel, Aaron Roth, and Zhiwei~Steven Wu.
\newblock Preventing fairness gerrymandering: Auditing and learning for
  subgroup fairness.
\newblock In Jennifer Dy and Andreas Krause, editors, {\em Proceedings of the
  35th International Conference on Machine Learning}, volume~80 of {\em
  Proceedings of Machine Learning Research}, pages 2564--2572. PMLR, 10--15 Jul
  2018.

\bibitem{kleinberg2018}
Jon~M. Kleinberg and Manish Raghavan.
\newblock Selection problems in the presence of implicit bias.
\newblock In Anna~R. Karlin, editor, {\em 9th Innovations in Theoretical
  Computer Science Conference, {ITCS} 2018, January 11-14, 2018, Cambridge, MA,
  {USA}}, volume~94 of {\em LIPIcs}, pages 33:1--33:17. Schloss Dagstuhl -
  Leibniz-Zentrum f{\"{u}}r Informatik, 2018.

\bibitem{CME19}
Caitlin Kuhlman, MaryAnn VanValkenburg, and Elke Rundensteiner.
\newblock Fare: Diagnostics for fair ranking using pairwise error metrics.
\newblock In {\em The World Wide Web Conference}, WWW '19, page 2936–2942,
  New York, NY, USA, 2019. Association for Computing Machinery.

\bibitem{MC21}
Anay Mehrotra and L.~Elisa Celis.
\newblock Mitigating bias in set selection with noisy protected attributes.
\newblock In {\em Proceedings of the 2021 ACM Conference on Fairness,
  Accountability, and Transparency}, FAccT '21, page 237–248, New York, NY,
  USA, 2021. Association for Computing Machinery.

\bibitem{MB2002}
Herve Moulin and Anna Bogomolnaia.
\newblock A simple random assignment problem with a unique solution.
\newblock {\em Economic Theory}, 19:623--636, 04 2002.

\bibitem{NCGW20}
Harikrishna Narasimhan, Andrew Cotter, Maya Gupta, and Serena Wang.
\newblock Pairwise fairness for ranking and regression.
\newblock {\em Proceedings of the AAAI Conference on Artificial Intelligence},
  34(04):5248--5255, Apr. 2020.

\bibitem{Roth82}
Alvin~E. Roth.
\newblock The economics of matching: Stability and incentives.
\newblock {\em Mathematics of Operations Research}, 7(4):617--628, 1982.

\bibitem{Roth84}
Alvin~E. Roth.
\newblock The evolution of the labor market for medical interns and residents:
  A case study in game theory.
\newblock {\em Journal of Political Economy}, 92(6):991--1016, 1984.

\bibitem{SKJ21}
Ashudeep Singh, David Kempe, and Thorsten Joachims.
\newblock Fairness in ranking under uncertainty.
\newblock In M.~Ranzato, A.~Beygelzimer, Y.~Dauphin, P.S. Liang, and J.~Wortman
  Vaughan, editors, {\em Advances in Neural Information Processing Systems},
  volume~34, pages 11896--11908. Curran Associates, Inc., 2021.

\bibitem{Zhou1990OnAC}
Lin Zhou.
\newblock On a conjecture by gale about one-sided matching problems.
\newblock {\em Journal of Economic Theory}, 52:123--135, 1990.

\end{thebibliography}
